\documentclass[11pt]{amsart}
\usepackage{amsfonts}
\usepackage{amssymb}
\usepackage[dvips]{graphics}
\usepackage{epsfig}
\usepackage{euscript}
\usepackage{color}

\usepackage{fancyvrb}

\usepackage[pdftex]{hyperref}

 \newtheorem{thm}{Theorem}[section]
 
 \newtheorem{lemma}[thm]{Lemma}
 \newtheorem{prop}[thm]{Proposition}
 \theoremstyle{definition}
 \newtheorem{defn}[thm]{Definition}
 \theoremstyle{remark}

 \numberwithin{equation}{section}
 
 \def\idtyty{{\mathchoice {\mathrm{1\mskip-4mu l}} {\mathrm{1\mskip-4mu l}} %
{\mathrm{1\mskip-4.5mu l}} {\mathrm{1\mskip-5mu l}}}}

\newcommand{\caA}{{\mathcal A}}
\newcommand{\caB}{{\mathcal B}}

\newcommand{\caD}{{\mathcal D}}

\newcommand{\caG}{{\mathcal G}}
\newcommand{\caH}{{\mathcal H}}

\newcommand{\caK}{{\mathcal K}}
\newcommand{\caL}{{\mathcal L}}
\newcommand{\caM}{{\mathcal M}}

\newcommand{\caO}{{\mathcal O}}

\newcommand{\caS}{{\mathcal S}}

\newcommand{\caU}{{\mathcal U}}
\newcommand{\caV}{{\mathcal V}}

\newcommand{\bbC}{{\mathbb C}}

\newcommand{\bbE}{{\mathbb E}}

\newcommand{\bbN}{{\mathbb N}}

\newcommand{\bbP}{{\mathbb P}}

\newcommand{\bbR}{{\mathbb R}}

\newcommand{\bbZ}{{\mathbb Z}}

\newcommand{\ie}{{\it i.e.\/} }
\newcommand{\eg}{{\it e.g.\/} }

\newcommand{\iu}{\mathrm{i}}

\newcommand{\str}{^{*}}
\newcommand{\ep}[1]{\mathrm{e}^{#1}}

\newcommand{\Tr}{\mathrm{Tr}}

\newcommand{\abs}[1]{\left\vert #1 \right\vert}
\newcommand{\braket}[2]{\left\langle #1 , #2\right\rangle}

\DeclareMathOperator*{\slim}{s-lim}

\newcommand{\spec}{\mathrm{spec}}

\newcommand{\be}{\begin{equation}}
\newcommand{\ee}{\end{equation}}
\newcommand{\bea}{\begin{eqnarray}}
\newcommand{\eea}{\end{eqnarray}}
\newcommand{\beann}{\begin{eqnarray*}}
\newcommand{\eeann}{\end{eqnarray*}}

\newcommand{\Rchain}{{[1,\infty)}}
\newcommand{\Lchain}{{(-\infty,0]}}

\begin{document}

\renewcommand{\thefootnote}{\fnsymbol{footnote}}
\title[Symmetric phases]{On gapped phases with a continuous symmetry \\ and boundary operators}

\author[S. Bachmann]{Sven Bachmann}
\address{Department of Mathematics\\
University of California, Davis\\
Davis, CA 95616, USA}
\email{svenbac@math.ucdavis.edu}

\author[B. Nachtergaele]{Bruno Nachtergaele}
\address{Department of Mathematics\\
University of California, Davis\\
Davis, CA 95616, USA}
\email{bxn@math.ucdavis.edu}

\date{\today }

\footnotetext{Copyright \copyright\ 2013 by the authors. This
paper may be reproduced, in its entirety, for non-commercial
purposes.}

\begin{abstract}
We discuss the role of compact symmetry groups, $G$, in the classification of gapped ground state phases of quantum spin systems. 
We consider two representations of $G$ on infinite subsystems. First, in arbitrary dimensions, we show that the ground state spaces 
of models within the same $G$-symmetric phase carry equivalent representations of the group for each finite or infinite sublattice on 
which they can be defined and on which they remain gapped. This includes infinite systems with boundaries or with 
non-trivial topologies. Second, for two classes of one-dimensional models, by two different methods, for $G=SU(2)$ in one, and 
$G\subset SU(d)$, in the other we construct explicitly an `excess spin' operator that implements rotations of half of the infinite chain 
on the GNS Hilbert space of the ground state of the full chain. Since this operator is constructed as the limit of a sequence of observables,
the representation itself is, in principle, experimentally observable. We claim that the corresponding unitary representation of $G$ is 
closely related to the representation found at the boundary of half-infinite chains. We conclude with determining the precise relation 
between the two representations for the class of frustration-free models with matrix product ground states. 
\end{abstract}

\maketitle

\begin{center} \emph{Dedicated to Herbert Spohn} \end{center}

\vspace{0.5cm}

\section{Introduction}\label{sec:Intro}

There are good reasons for the intense activity around gapped ground states phases of quantum lattice models. Theoretically as well 
as experimentally it has become clear that ground states of many models in condensed matter physics often have interesting 
structural properties even in the absence of long-range order. This structure is sometimes referred to as 
{\em topological order} because of striking topological effects found in some important examples \cite{Wen:1995}. 
A quantum ground state phase is 
loosely defined as a family of models (Hamiltonians) of which the ground states have qualitatively similar properties. The recent 
literature provides ample support for the following statement: for a family of Hamiltonians depending on a continuous parameter
in some range, $\lambda\in I$, its ground states represent the same quantum phase if there is a positive lower bound for the spectral 
gap above the ground state energy independent of $\lambda\in I$. In the opposite case, a {\em quantum phase transition} may occur at a 
critical value of the parameter if some of the ground state properties undergo qualitative changes, and the spectral gap above the 
ground state necessarily vanishes at this critical value.

Although quantum phase transitions were already described a long time ago, for example in the XY model~\cite{Lieb:1961vs} or in the 
quantum Ising model in a transverse field~\cite{Pfeuty:1970wq}, and in the monograph \cite{Sachdev:1999}, the notion of a gapped 
ground state phase has only recently been thoroughly investigated~\cite{Hastings:2005cs, Kitaev:2009,Chen:2010gb, 
Chen:2011iq, Schuch:2011ve, Bachmann:2011kw, Bachmann:2012uu}. The classification of symmetry-protected phases has been the
subject of several studies. The result of such a classification will of course depend on the notion of equivalence that is employed. In most of the
literature the manipulations preserving the equivalence of systems includes embedding the Hilbert space of the original system as a subspace of
the Hilbert space of another system (extending the range of degrees of freedom), and also tensoring with another system (adding degrees of freedom)   
\cite{Hastings:2013wp}. This approach has led to appealing results \cite{Kitaev:2009,Schnyder:2009,Duivenvoorden:2013,Chen:2013}. In this article, 
however, we do not compare systems with different Hilbert spaces. What we find then is a stronger invariant, meaning that our results have a greater 
power to predict the existence of a quantum phase transition, i.e., the existence of at least one critical point in any parameter region between two 
values where the system is in two different gapped phases distinguished by different values of the invariant.

In the one-dimensional classes of examples we work out explicitly (see the excess spin operators in Section \ref{sec:Excess}), our invariant is 
the representation generated by unitary string operators, closely related to the string order parameters \cite{denNijs:1989vl}, acting on the ground 
states of a half-infinite chain. Remarkably, these unitaries can be approximated by local observables and, hence, are in principle experimentally 
accessible. In a very distinct way {\em non-local} unitaries can also be put to use to unravel the structure of certain symmetry-protected phases, as 
was first done by Kennedy and Tasaki in \cite{Kennedy:1992ta}. See \cite{Duivenvoorden:2013b,Else:2013} for several generalizations.

An important consequence of the existence of a smooth interpolating path connecting two Hamiltonians along which the spectral gap 
does not vanish, is the local automorphic equivalence \cite{Bachmann:2011kw} of the ground state spaces of the models along the path.
In \cite{Bachmann:2012bf, Bachmann:2012uu} we showed the existence of such a path connecting a class of generalizations of the 
AKLT chain \cite{Affleck:1988vr, Tu:2008bq} to a class of toy models we called PVBS Hamiltonians, and emphasized the role of edge states.
Further physical and mathematical properties shared by models in the same phase remain largely conjectural. For example, the definition of a 
gapped ground state phase is expected to encapsulate the common structure of entanglement within the system, as embodied by the so-called 
entanglement spectrum in these systems which are expected to satisfy an area law for the entanglement entropy~\cite{Turner:2011,Michalakis:2012wq}. 
More generally, models within the same phase should exhibit the same topological order. Automorphically equivalent models have the same ground 
state degeneracy on any lattice they can be defined on, and if a symmetry is preserved along the path, they exhibit the same pattern of symmetry 
breaking. Moreover, models in the same phase are even expected to share a similar structure of their low-lying excitations, as for example quasi-
particles with anyonic properties in the case of topologically ordered phases~\cite{Kitaev:2003ul}. The mathematical study of these question 
has only just begun \cite{Naaijkens:2011,Levin:2012,Naaijkens:2013, Haegeman:2013}.

In this article, we study gapped quantum phases with symmetry, in particular a local continuous symmetry such as the invariance under spin
rotations of antiferromagnetic spin chains.  The assumption of a non-vanishing gap implies, under rather general conditions, that 
the symmetry is necessarily unbroken and consequently that the ground state does not have the associated long-range 
order \cite{Landau:1981gg}. 

Concretely, we will prove some implications of the automorphic equivalence of $G$-invariant models, where $G$ is a group of local symmetries
of the Hamiltonian.
That is, we assume that the $G$-invariance is preserved along the paths connecting different models in what could be called the same
{\em symmetry-protected gapped phase}, in analogy with the {\em symmetry-protected topological order} that may be present and detectable in such a 
situation \cite{Pollmann:2012wv, Pollmann:2012,Haegeman:2012bx, Chen:2013}. It has already been argued in~\cite{Bachmann:2012bf} that the bulk 
ground state is 
not sufficient to describe the full range of possible phases and phase transitions even in one dimension, and that the structure of edge states carries additional 
information. In the presence of a symmetry which remains unbroken in the thermodynamic limit, we prove that these edge modes, namely the 
additional ground states that arise whenever a boundary is present, carry equivalent representations of $G$ within a phase. The argument is 
valid in any dimension and for arbitrary lattices including domains with holes or non-trivial topology. The set of these edge representations can therefore 
be used as a new general invariant of a gapped ground state phase.

The results of Section~\ref{sec:GTransport} are not concerned with a particular model, but are statements about properties shared by all models within a symmetric gapped ground state phase, thereby clarifying the meaning and implications of the very notion of a phase.

In one dimension and for the special case $G\subset SU(2)$, the abstract result can be supported by the construction of an explicit 
`excess spin' operator generating the rotations of the half chain in the GNS Hilbert space of the bulk state. This experimentally accessible observable provides a way to measure the abstract invariant of a phase in a particular system. For models in the Haldane phase, namely for integer spin antiferromagnets with a unique gapped ground state in the thermodynamic limit and exponential decay 
of correlations, the existence of a half-integer edge spin is readily apparent in the AKLT model and was experimentally observed in~\cite{Hagiwara:1990wk}. We show by construction that the excess spin can be approximated by local observables in two particular cases: for frustration free models having a finitely correlated ground state (in this situation our approach works for any $G\subset SU(d)$), and for a class of antiferromagnetic chains whose ground state has a classical random loop representation. Isotropic half-integer antiferromagnetic spin chains are expected to either belong to a a gapless phase or to a dimerized phase, breaking translation invariance~\cite{Affleck:1986wu}. The results of this article can be applied to dimerized phases but we will not discuss the implications in detail here. Similarly, the XY chain does generically not carry an interesting local symmetry except in the isotropic case in which it becomes gapless.

In Section~\ref{sec:boundary-bulk} we show how the abstract boundary representation and the excess spin operator 
are related in the case of frustration free models. In the context of frustration free spin chains with matrix product ground states we 
prove an exact relation between these representations and, as a consequence, show that the representation that characterizes the 
edge ground state space can be determined from the correlations in the bulk. This can be seen as a simple instance of \emph{bulk-edge correspondence} 
in the setting of quantum spin models, similar to the bulk-edge correspondence in quantum Hall systems or topological insulators.


\section{Symmetric gapped phases}\label{sec:GTransport}

We consider a quantum spin system with local Hilbert space $\caH^x=\bbC^d$ at each site $x\in\Gamma$. Here, $\Gamma$
is a countable metric space, e.g., a finite-dimensional lattice with the lattice distance as the metric.
The local algebra of observables on any finite subset $\Lambda\subset\Gamma$ is the full matrix algebra 
$\caA^{\Lambda} := \caL(\otimes_{x\in\Lambda}\caH^x)$ and, for $\Lambda_1\subset \Lambda_2$, we consider 
$\caA^{\Lambda_1}$ as a subalgebra of $ \caA^{\Lambda_2}$, in the natural way. For infinite $\Gamma$, the algebra 
of quasi-local observables, $\caA^{\Gamma}$, is the norm completion of $\bigcup_{\Lambda\subset \Gamma}\caA^{\Lambda}$,
where the union is over the finite subsets of $\Gamma$. By the limit $\Lambda\to\Gamma$, we mean the limit along an increasing and absorbing sequence of finite subsets $\Lambda_n\subset\Gamma$.

We shall require some local structure for the set $\Gamma$. First a uniform bound on the rate at which balls grow: There exist numbers $\kappa >0$ and $\nu >0$ for which $\sup_{x \in \Gamma} |B_r(x) | \leq \kappa r^{\nu}$, where $|B_r(x)|$ is the cardinality of the ball centered at $x$ of radius $r$. Second, we assume that $\Gamma$ has some underlying `integrable' structure. There exists a non-increasing, real-valued function $F: [0, \infty) \to (0, \infty)$ that satisfies a uniform summability condition
\begin{equation*}
\sup_{x \in \Gamma} \sum_{y \in \Gamma} F(d(x,y)) < \infty,
\end{equation*}
and a convolution condition, namely there is a constant $C_F$ such that
\begin{equation*}
\sum_{z \in \Gamma} F(d(x,z)) F(d(z,y) \leq C_F F(d(x,y)),
\end{equation*}
for any $x,y\in\Gamma$. Finally, we will assume that $F$ does not decay to zero faster than polynomially (see~\cite{Bachmann:2011kw} for a possible specific condition).

In this article, we shall consider only compact Lie groups $G$ as local symmetries. Although this choice corresponds to the examples we have in 
mind, \eg quantum mechanical rotations $G=SU(2)$, we note that not all arguments require compactness. We assume that $\caH^x$ carries a
 unitary representation of $G$, $g\mapsto U^x_g\in SU(d)$, which determines the adjoint action $\theta^x_g := \mathrm{Ad}_{U^x_g}$ on $\caA^{\{x\}}$. 
 It naturally extends to an action $\theta^\Lambda_g = \otimes_{x\in\Lambda}\theta^x_g$ over $\caA^{\Lambda}$ and further to $\theta_g^\Gamma$ on 
 the quasi-local algebras $\caA^\Gamma$. Concretely, 
for finite subsets $\Lambda$,
\begin{equation*}
U^\Lambda_g = \otimes_{x\in\Lambda}U^x_g,
\end{equation*}
but the automorphism $\theta_g$ in general fails to remain inner on the quasi-local algebra.

By a local Hamiltonian of the spin chain, we mean a map $\Lambda\subset\Gamma \rightarrow H^\Lambda = (H^\Lambda)\str \in\caA^\Lambda$ defined on finite subsets of the lattice and of the form
\begin{equation}
H^\Lambda = \sum_{X\subset\Lambda}\Phi(X),
\end{equation}
where $\Phi(X) = \Phi(X)\str = \caA^X$ represents the interactions among the spins located in $X$. The following condition expresses the decay of the interaction at large distances,
\begin{equation}\label{Norm}
\Vert\Phi\Vert_\mu:=\sup_{x,y \in \Gamma} \frac{\ep{\mu d(x,y)}}{F(d(x,y))} 
\sum_{\stackrel{Z \subset \Gamma:}{x,y \in Z}} \| \Phi_Z \| < \infty,
\end{equation}
for a $\mu>0$. With this, the Heisenberg dynamics corresponding to $H^\Lambda$, $\tau^\Lambda_t(A) = \exp(-\iu t H^\Lambda) A \exp(\iu t H^\Lambda)$ satisfies a Lieb-Robinson bound
\begin{equation*}
\Vert [\tau^\Lambda_t(A),B] \Vert \leq C(A,B)\ep{-\mu (d(Z_1,Z_2)-v(\Phi)t)}
\end{equation*}
where $A\in\caA^{Z_1}$, $B\in\caA^{Z_2}$ and $d(Z_1,Z_2)$ is the distance between the supports of $A$ and $B$. See \eg~\cite{Nachtergaele:2006bh} for explicit expressions for the constant $C(A,B)$ and the Lieb-Robinson velocity $v(\Phi)$. 

Besides the dynamical condition expressed by the Lieb-Robinson bound, the Hamiltonians we shall consider here also satisfy a spectral condition, namely that they are gapped. Precisely,
\begin{defn}\label{def:gapped}
Let $\lambda_0(\Lambda) = \inf\spec(H^\Lambda)$. The model is gapped if there exists a constant $\gamma >0$ and a family $0\leq \epsilon_\Lambda<\gamma$ such that $\lim_{\Lambda\to\Gamma} \epsilon_\Lambda =0$, and
\begin{equation*}
\spec(H^\Lambda) \cap (\lambda_0(\Lambda) + \epsilon_\Lambda, \lambda_0(\Lambda) + \gamma) =\emptyset, \quad \mbox{for all } \Lambda.
\end{equation*}
\end{defn} 
We shall refer to the eigenstates corresponding to the spectral patch $[\lambda_0(\Lambda),\lambda_0(\Lambda) + \epsilon_\Lambda]$ as ground states and denote the set of such states as $\caG^\Lambda$, despite the fact that they become strictly degenerate only in the limit $\Lambda\to\Gamma$. More importantly, these low-lying excitations are isolated from the rest of the spectrum by a positive gap, uniformly in $\Lambda$. Among the models that we have in mind is the spin-$1$ Heisenberg antiferromagnet in one dimension, which is known numerically to exhibit exactly the above behavior, see~\cite{Kennedy:1999ty}. There, the eigenvalue splitting $\epsilon_\Lambda$ between the spin singlet and the triplet decays exponentially.

Finally, we define precisely what we mean by a gapped ground state phase carrying a symmetry for a quantum spin system having all of the above structure. We consider a differentiable family of interactions $\Phi(s)$ to which we associate $\partial\Phi(Z,s)=\vert Z\vert \Phi'(Z,s)$,
where the prime denotes differentiation with respect to $s$, for each finite $Z \subset \Gamma$, and assume that
\begin{equation*}
\Vert\partial\Phi\Vert_\mu < \infty,
\end{equation*}
as in~(\ref{Norm}), uniformly in $s$.

\begin{defn}\label{def:phase}
Let $G$ be a group. Two gapped local Hamiltonians $H_0$ and $H_1$ are in the same $G$-symmetric gapped ground state phase if
\begin{enumerate}
\item There exists a smooth family of gapped local Hamiltonians $[0,1]\ni s\mapsto H(s)$ such that $H_0 = H(0)$ and $H_1 = H(1)$;
\item The interaction $\Phi(s)$ is invariant under the action of $G$,
\begin{equation*}
\theta_g(\Phi(Z,s)) = \Phi(Z,s)
\end{equation*}
for all $g\in G$, $Z \subset \Gamma$, and $s\in[0,1]$.
\end{enumerate}
\end{defn} 
Let $\caS^\Gamma(s)\subset(\caA^\Gamma)\str$ be the set of ground states of the model on $\Gamma$, \ie the set of states on $\caA^\Gamma$ that are accumulation points of the functionals $\braket{\psi^\Lambda}{\cdot \ \psi^\Lambda}$, where $\psi^\Lambda$ is a normalized state in $\caG^\Lambda(s)$.

As shown in~\cite{Bachmann:2011kw}, the conditions of Definition~\ref{def:phase} allow for the explicit construction of a cocycle of quasi-local automorphisms $\alpha^\Gamma_{s,t}$ of $\caA^\Gamma$ such that
\begin{equation} \label{GSphase}
\caS^\Gamma(t) = \caS^\Gamma(s) \circ \alpha^\Gamma_{s,t},
\end{equation}
which justifies the denomination of a ground state phase. Precisely, $\alpha^\Gamma_{s,t}$ is obtained as the limit as $\Lambda\to\Gamma$ of a unitary conjugation, where the unitaries are generated by the $s$-dependent
\begin{equation*}
D^\Lambda(s) = \int_\bbR d\xi w(\xi)\int_0^\xi d\zeta\ep{-\iu \zeta H^\Lambda(s)} \frac{d}{ds}H^\Lambda(s) \ep{\iu \zeta H^\Lambda(s)}.
\end{equation*}
Here, $w$ is a probability density decaying faster than any polynomial at infinity and whose Fourier transform is compactly supported in $[-\gamma,\gamma]$. As a consequence, $\alpha^\Gamma_{s,t}$ is norm preserving,
\begin{equation} \label{NormPreserving}
\Vert \alpha^\Gamma_{s,t}(A) \Vert = \Vert A \Vert,
\end{equation}
and the $G$-invariance of the interaction carries over to the automorphisms $\alpha^\Gamma_{s,t}$, namely
\begin{equation} \label{AutomorphicInvariance}
\alpha^\Gamma_{s,t} \circ \theta^\Gamma_g = \theta^\Gamma_g\circ\alpha^\Gamma_{s,t},
\end{equation}
for all $s,t\in[0,1]$ and $g\in G$.

In the sequel, we shall mostly be interested in the dual action of the automorphisms $\alpha^\Gamma_{s,t}$ and $\theta^\Gamma_g$, that we denote by $\beta^\Gamma_{s,t}$ and $\Theta^\Gamma_g$, namely
\begin{equation*}
\beta^\Gamma_{s,t}(l)(A) := l\left(\alpha^\Gamma_{s,t}(A)\right) ,\qquad \Theta^\Gamma_g(l)(A) := l\left(\theta^\Gamma_g(A)\right),
\end{equation*}
for any $l\in(\caA^\Gamma)\str$. 

We are now ready to state the first main result of this article.
\begin{thm}\label{thm:G-rep}
Let $H_0$ and $H_1$ be in the same $G$-symmetric gapped ground state phase. For $i=0,1$, let $\bar{\caS}^\Gamma_i$ be the vector spaces spanned by elements of $\caS^\Gamma_i$. Then, for each $\Gamma$,
\begin{enumerate}
\item $\mathrm{dim}(\bar{\caS}^\Gamma_0) = \mathrm{dim}(\bar{\caS}^\Gamma_1)$, in the sense that if one is finite-dimensional, the other is too and with the same dimension,
\item $\bar{\caS}^\Gamma_i$ carry a representation $\Theta^\Gamma_{g;i}$ of $G$,
\item $\Theta^\Gamma_{g;0}$ and $\Theta^\Gamma_{g;1}$ are equivalent.
\end{enumerate}
\end{thm}

From now on, $H(s)$, respectively $\Phi(s)$, will denote a smooth family of $G$-invariant gapped Hamiltonians, respectively interactions, interpolating from $H_0$ to $H_1$ and satisfying all requirements described above. The key of the proof of the theorem is the following observation about the maps $\beta^\Gamma_{s,t}$ and their restriction to the ground state spaces.
\begin{prop}\label{prop:homeomorphism}
Equip $\caS^\Gamma(s)$ with the norm topology. Then $\caS^\Gamma(s_1)$ and $\caS^\Gamma(s_2)$ are homeomorphic for any $s_1,s_2\in[0,1]$.
\end{prop}
\begin{proof}
We claim that $\beta^\Gamma_{s_1,s_2}\upharpoonright_{\caS^\Gamma(s_1)}$ is the sought homeomorphism. By construction, its range is $\caS^\Gamma(s_2)$. Since it is a linear map, $\beta^\Gamma_{s_1,s_2}(\omega) = \omega \circ \alpha^\Gamma_{s_1,s_2}$, continuity follows from boundedness:
\begin{equation*}
\Vert \beta^\Gamma_{s_1,s_2}(\omega) \Vert = \sup_{\Vert A \Vert=1}\Vert \omega(\alpha^\Gamma_{s_1,s_2}(A)) \Vert = \Vert \omega \Vert
\end{equation*}
by~(\ref{NormPreserving}). Since $\alpha^\Gamma$ is a cocycle, its dual is invertible with $(\beta^\Gamma_{s_1,s_2})^{-1} = \beta^\Gamma_{s_2,s_1}$. Hence the inverse is also continuous.
\end{proof}

Although we shall not use this in the sequel, we note that $\beta^\Gamma_{s_1,s_2}$ is also weakly continuous in the `time variable'. Indeed, by the cocycle property,
\begin{equation*}
\abs{\left(\beta^\Gamma_{s_1,s_3}(\omega) - \beta^\Gamma_{s_1,s_2}(\omega)\right)(A)}
= \abs{\beta^\Gamma_{s_1,s_2}(\omega)\left(\alpha^\Gamma_{s_1,s_3-s_2}(A)-A\right)}\leq \Vert \alpha_{s_1,s_3-s_2}^\Gamma(A)-A \Vert,
\end{equation*}
and the conclusion follows from the continuity of $t\mapsto \alpha_{s,t}^\Gamma$.

\begin{proof}[Proof of Theorem~\ref{thm:G-rep}]
i. It suffices to note that $\beta^\Gamma_{s_1,s_2}\upharpoonright_{\bar{\caS}^\Gamma(s_1)}$ is a linear bijective map from $\bar{\caS}^\Gamma(s_1)$ to $\bar{\caS}^\Gamma(s_1)$, as shown in the proof above. \\
ii. We drop the index $i$, and show that $\caS^\Gamma$ is invariant under the action of $G$. Since the Hamiltonian $H^\Lambda$ commutes with the representation $U^\Lambda_g$, $\psi^\Lambda\in\caG^\Lambda$ implies $U^\Lambda_g\psi^\Lambda\in\caG^\Lambda$. This implies that the set of functionals $\omega_{\psi^\Lambda}:=\braket{\psi^\Lambda}{\cdot\,\psi^\Lambda}/\Vert\psi^\Lambda\Vert^2$ is invariant under $\Theta^\Lambda_g$. By going to subsequences, any $\omega^\Gamma\in\caS^\Gamma$ is the weak-* limit of a $\omega^\Lambda\in\caS^{\Lambda}$. Hence,
\begin{equation*}
\Theta^\Gamma_g(\omega)(A) = \omega^\Gamma(\theta^\Gamma_g(A)) = \lim_{\Lambda\to\Gamma}\omega^\Lambda(\theta^\Gamma_g(A))
\end{equation*}
for any $A\in\caA^{\Lambda}$. But by the first remarks, $\omega^\Lambda\circ \theta^\Gamma_g \in \caS_{\Lambda}$. Therefore, $\Theta^\Gamma_g\caS^\Gamma\subset\caS^\Gamma$. Extending this action by linearity to $\bar{\caS}^\Gamma$ yields a representation. \\
iii. Let $\omega\in\caS^\Gamma_0$. By definition, we have that
\begin{equation*}
\beta^\Gamma_{0,1}\left(\Theta^\Gamma_{g,0}(\omega)\right)(A)
= \omega\left(\theta^\Gamma_g\circ \alpha^\Gamma_{0,1}(A)\right)
= \omega\left(\alpha^\Gamma_{0,1}\circ\theta^\Gamma_g(A)\right)
= \Theta^\Gamma_{g,1}\left(\beta^\Gamma_{0,1}(\omega)\right)(A)
\end{equation*}
where we used the covariance~(\ref{AutomorphicInvariance}) in the second equality. This relation and 
Proposition~\ref{prop:homeomorphism} imply that $\beta^\Gamma_{0,1}$ is an intertwining isomorphism between the 
representations $\bar{\caS}^\Gamma(0)$ and $\bar{\caS}^\Gamma(1)$.
\end{proof}
The one-dimensional VBS models, such as the AKLT model \cite{Affleck:1988vr} and its many generalizations, provide a rich class of examples for 
which the representation $\Theta^\Gamma$, with $\Gamma$ the left or right half-infinite chain, is explicitly known. For examples, in the case of the $SU(2)$-invariant spin-1 AKLT chain, the representation is the adjoint 
representation of the spin-1/2 representation, which is equivalent to the direct sum of a singlet and a triplet.


\section{The excess spin operator}\label{sec:Excess}

As can be readily seen,  the results of the previous section rely neither on the Lie structure nor on the compactness of $G$, and the representations 
$\Theta_g^\Gamma$ that have been discussed so far have been of a general group, in any dimension and at the abstract level of sets of ground states. 
Now, we restrict our attention to one-dimensional systems and shall use both the Lie group structure and the fact that the representations form connected 
subgroups of $SU(d)$. The possible infinitely extended lattices are, beside $\bbZ$, the two half-infinite chains with one boundary $\Rchain$ and $\Lchain$. 
We will write $\caA^\Gamma$ whenever we refer to any of the quasi-local algebras. In this more specific case, it is possible to construct an explicit unitary 
implementation of the group action on the half of the infinite chains for two classes of symmetric models.

The unitaries we shall construct find their origin in the special case of integer spin representations of $G=SU(2)$, 
for antiferromagnetic spin chains. There, despite the absence of long-range order in the classical sense, it was observed early, \eg in~\cite{denNijs:1989vl, Girvin:1989uk, Kennedy:1992ta}, that the following `string order parameter'
\begin{equation*}
O_{x,y} = (-1)^{y-x} \omega\left(S^x \ep{\iu\pi\sum_{j=x+1}^{y-1}S^j} S^y\right)
\end{equation*}
is a good characterization of the dilute N\'eel order in the ground state of the AKLT chain. {\em Dilute N\'eel order} refers to property that in any finite 
spin configuration that appears with non-zero probability in the ground state, $+1$'s and $-1$'s strictly alternate (hence, {\em N\'eel} order), but the 
$\pm1$'s are separated by a random number of $0$'s (hence, {\em dilute}). It was already pointed out in~\cite{Aizenman:1994th} that the exponential 
of the formal sum $\sum_{j=1}^\infty S^j$ which generates rotations of the half-infinite chain can be expected to converge in the GNS representation of a bulk state exhibiting sufficient decay of correlations, and that it is a robust generalization of the string order parameter. Moreover, this operator clarifies the origin of the edge states carrying a half-integer spin that have been observed in some spin-$1$ chains~\cite{Hagiwara:1990wk}, and is related to the representations that have been discussed in the previous section.

There are two classes of models for which such boundary operators can be shown to exist. The first class is that of models whose ground states have a valence bond representation \cite{Affleck:1988vr}, for which the boundary operator arises quite naturally and for arbitrary Lie groups. This is the family of frustration free chains that can also be described as finitely correlated states \cite{Fannes:1992vq}. The second family of models carry a $G=SU(2)$ symmetry and have ground states which have a stochastic-geometric interpretation, see~\cite{Aizenman:1994th, Nachtergaele:1993wk, Ueltschi:2013uy}. We now address both situations.

\subsection{Excess spin on finitely correlated chains}\label{sub:FCS}

We start by recalling some basic facts about finitely correlated states -- also known as matrix product states -- that are invariant under a gauge group, see~\cite{Fannes:1992vq}. We consider a purely generated finitely correlated state $\omega$ given 
in its `minimal representation' (not to be confused with the algebra representations we discussed so far), see~\cite{Fannes:1994511}. The auxiliary algebra is $\caB = \caM_k$, and $\omega$ is generated by the triple $(\bbE,\rho,\idtyty)$, where $\bbE_A(b) = V \str (A\otimes b) V$ and $V:\bbC^k\rightarrow \bbC^d\otimes\bbC^k$.

If $\omega$ is invariant under the local gauge group $G$, then both $\bbC^k$ and $\bbC^d$ carry a unitary representation of $G$ and $V$ can be chosen to be the following natural isometric intertwiner
\begin{equation}\label{intertwine}
(U_g\otimes u_g) V = V u_g,\quad \text{and hence}\quad V\str (U_g\otimes u_g) V = u_g.
\end{equation}
In other words, $u_g$ is an eigenvector of eigenvalue $1$ for $\bbE_{U_g}$. For the neutral element in $G$, this implies that $\bbE_\idtyty(\idtyty) = \idtyty$.

If the identity is the unique eigenvector for the eigenvalue $1$ and
\begin{equation}\label{FCS_Spectrum}
\lambda_e := \max (\abs{\lambda}:\lambda\in\mathrm{spec}(\bbE_{\idtyty})\setminus\{1\}) <1,
\end{equation}
then the ground state $\omega$ is translation invariant with exponential decay of correlations
\begin{equation*}
\abs{\omega(A\str B)-\omega(A\str)\omega(B)}\leq C \lambda_e^l,
\end{equation*}
where $l=\mathrm{dist}(\mathrm{supp}(A),\mathrm{supp}(B))$. Moreover, $1$ is also a simple eigenvalue of the transpose map $\bbE_\idtyty^t$, and the corresponding eigenvector can
be chosen to be a state $\rho$. Explicitly, $\rho$ is the unique density matrix such that
\begin{equation}\label{E1}
\Tr \left(\rho \bbE_\idtyty(b)\right) = \Tr (\rho b),
\end{equation}
for all $b\in\caB$. By continuity and~(\ref{intertwine}), for $g$ in an open neigborhood of $e$, 1 is also a simple eigenvalue of $\bbE_{U_g}$, 
and a gap condition similar to~(\ref{FCS_Spectrum}) will hold for $\bbE_{U_g}$ and $\bbE_{U_g}^t$. Using~(\ref{intertwine}) again and the unitarity of the representation $U_g$, we also have that for any $b\in\caB$
\begin{equation*}
\Tr\left(\rho u_g\str \bbE_{U_g}(b)\right) 
= \Tr\left(\rho V\str (U_g\str\otimes u_g\str)(U_g\otimes b) V \right) 
= \Tr\left(\rho \bbE_{\idtyty}(u_g\str b)\right)
= \Tr\left(\rho u_g\str b\right)
\end{equation*}
where the last equality follows from~(\ref{E1}). Hence, $\rho u_g^*$ is the eigenvector of $\bbE_{U_g}^t$ with eigenvalue 1.

Without loss of clarity but slightly abusing notation, we can assume that there is an element $S=S^*$ of the representation of the Lie algebra such that $U_g=\exp(igS)$, for $g\in\bbR$. We now turn our attention to the (formal) boundary operators $\sum_{x = 1}^{\infty} S^x$ and $\sum_{x = -\infty}^{0} S^x$. It is convenient to define them as limits as $L\to\infty$ of the following local approximations
\begin{equation}\label{SpmL}
S^+(L) = \sum_{x = 1}^{L^2} f_L(x-1) S^x,\qquad S^-(L) = \sum_{x = -L^2+1}^0 f_L(-x) S^x,
\end{equation}
where $f_L:\bbZ^+\to \bbR$ is given by
\begin{equation}\label{f_L}
f_L(mL+n)=1- m/L, \mbox{ for } m,n\in [0,L-1], \mbox{ and } f(x)=0, \mbox{ for } x\geq L^2.
\end{equation}
Correspondingly, we shall write
\begin{equation} \label{Ug+f}
U_g^+(L)= \exp ( igS^+(L)).
\end{equation}
\begin{thm}\label{thm:ExcessSpin_FCS}
Let $\omega$ be an $G$-invariant finitely correlated state generated by the intertwiner~$V$ as above, and let $(\caH_\omega,\pi_\omega,\Omega_\omega)$ be its GNS representation. Assume that~(\ref{FCS_Spectrum}) holds. Then the strong limits
\begin{equation*}
U^+_g = \slim_{L\to\infty} \ep{\iu g \cdot \pi_\omega\left(S^+(L)\right)},\qquad 
U^-_g = \slim_{L\to\infty} \ep{\iu g \cdot \pi_\omega\left(S^-(L)\right)}
\end{equation*}
exist on $\caH_\omega$ for all $g\in G$ and define infinite-dimensional, strongly continuous representations of $G$. Moreover,
\begin{equation*}
U^+_g\in\pi_\omega\left(\caA^\Lchain\right)',\qquad U^-_g \in\pi_\omega\left(\caA^\Rchain\right)',
\end{equation*}
where $\caU'$ denotes the commutant of the algebra $\caU$.
\end{thm}
Before commencing the proof of the theorem, we make a technical observation in the form of the following lemma.
\begin{lemma}\label{lma:UnitaryCauchy}
Let $(U_n(\cdot))_{n\in\bbN}$ be a sequence of unitary operators on $\caH$ such that
\begin{equation}\label{Cauchy_Gen_1}
\lim_{\min(n,m)\to\infty}\braket{\phi}{\left(1-U_n\str U_m\right)\psi} = 0,\qquad \lim_{\min(n,m)\to\infty}\braket{\phi}{\left(1-U_n U_m\str\right)\psi} = 0,
\end{equation}
for all $\phi,\psi$ in a dense subset $\caD\subset\caH$. There exists a unitary operator $U$ such that $\slim_{n\to\infty} U_n = U$.
\end{lemma}
\begin{proof}
By Cauchy-Schwarz,
\begin{equation*}
\abs{\braket{\phi}{(U_n - U_m)\psi}}^2 
\leq \Vert\phi\Vert^2 \braket{\psi}{(2-U_n\str U_m-U_m\str U_n)\psi},
\end{equation*}
so that (\ref{Cauchy_Gen_1}) implies the existence of a bounded sesquilinear form $Q$ such that $\braket{\phi}{U_n\psi}\to Q(\phi,\psi)$. This in turn
implies the existence of unique bounded linear operator $U$ on $\caH$ such that $\braket{\phi}{U\psi}= Q(\phi,\psi)$, and $U_n\to U$ in the weak 
operator topology. Similarly, $U_n\str$ converges to $U\str$. Unitarity of the limit follows from~(\ref{Cauchy_Gen_1}) again as
\begin{equation*}
\braket{\phi}{U\str U\psi} = \lim_{n\to\infty}\lim_{m\to\infty}\braket{U_n\phi}{U_m\psi} = \braket{\phi}{\psi}.
\end{equation*}
It remains to note that strong and weak convergences are equivalent for unitary operators.
\end{proof}
\begin{proof}[Proof of Theorem~\ref{thm:ExcessSpin_FCS}]
In this proof, we consider group elements with $g\in(-s,s)\subset\bbR$ for some $s>0$, or equivalently arbitrary group elements belonging 
to an open neighborhood of the neutral element of $G$, such that $\bbE_{U_g}$ has the properties mentioned following~(\ref{FCS_Spectrum}).

Because of the spectral condition, there exists a rank 1 operator $P_g$ on $\caB$, with $\Vert P_g\Vert =1$,
and constants $C>0$ and $\lambda\in (0,1)$, such that
\begin{equation}\label{ConUg}
\Vert \bbE_{U_g}^{(n)}-P_g \Vert \leq C\lambda^n,
\end{equation}
for all $g\in (-s,s)$ and $n\geq 0$. By the previous discussion we have for any $b\in\caB$
\begin{equation*}
P_g (b) = \Tr (\rho u^*_g b) u_g.
\end{equation*}
We further define the rank 1 map $Q_g$ on $\caB$ by
\begin{equation*}
Q_g (b) = \Tr(\rho b) u_g.
\end{equation*}
First, we prove the following estimate:
\begin{equation*}
\left\Vert \bbE^{(L^2)}_{U_{g}^+(L)}-Q_g\right\Vert \leq C_1/L,
\end{equation*}
where $U_{g}^+(L)$ is defined in~(\ref{Ug+f}) and $C_1$ is a constant independent of $g$ and $L$.

We shall use that
\begin{equation*}
\bbE^{(L^2)}_{U_{g}^+(L)} = \prod_{x=0}^{L^2-1} \bbE_{U_{gf_L(x)}}
\end{equation*}
where the product is ordered left-to-right, and $f_L$ was introduced in~(\ref{f_L}). Note that $f_L$ is constant on blocks of length $L$. The first step is to use a telescopic sum of products and~(\ref{ConUg}) to obtain the estimate
\begin{equation*}
\bigg\Vert\prod_{x=0}^{L^2-1} \bbE_{U_{gf_L(x)}}-\prod_{n=L}^{1}P_{ng/L}\bigg\Vert \leq  L C\lambda^L.
\end{equation*}
Next, we analyze the product of rank 1 operators as follows:
\begin{equation}\label{rank1product}
\left(\prod_{n=L}^{1}P_{ng/L}\right)(b)
= \left( \prod_{n=L}^2 \Tr \left(\rho u^*_{ng/L} u_{(n-1)g/L}\right)\right)\Tr \left(\rho u^*_{g/L} b\right) u_g.
\end{equation}
Since $\det u_g =1$, and $\rho$ commutes with $u_g$, we have $\Tr \rho S=0$ by Schur's lemma. Using this and the group property for the 
representation $u_g$, we find that $\Tr \rho u^*_g$ is quadratic in $g$ for small $g$. Explicitly, 
\begin{equation*}
\Tr (\rho u^*_{g/L}) -1 = - (g/L)^2 \Tr(\rho S^2) + \caO((g/L)^3).
\end{equation*}
This estimate implies that there exists a constant $C_2$ such that
\begin{equation*}
\left\vert\left(\prod_{n=L}^2 \Tr \rho u^*_{ng/L} u_{(n-1)g/L}\right)-1\right\vert\leq \left( 1+ C_2 (g/L)^2\right)^{L-1} -1\leq \ep{C_2\frac{g^2}{L-1}}-1.
\end{equation*}
Similarly for the last factor in~(\ref{rank1product}), we have for another constant $C_3$
\begin{equation*}
\abs{\Tr (\rho u\str_{g/L} b) -\Tr (\rho b)} \leq C_3\Vert b\Vert \abs{g}/L,
\end{equation*}
and therefore $\abs{\Tr (\rho u\str_{g/L} b) u_g  - Q_g(b)}\leq C_3 \Vert b\Vert \abs{g}/L$.
Putting these estimates together, we obtain 
\begin{align*}
\left\Vert \bbE^{(L^2)}_{U_{g}^+(L)} -Q_g\right\Vert 
&\leq \bigg\Vert\prod_{x=0}^{L^2-1} \bbE_{U_{gf_L(x)}}-\prod_{n=L}^{1}P_{ng/L}\bigg\Vert
+ \bigg\Vert \prod_{n=L}^{1}P_{ng/L} - Q_g\bigg\Vert\\
&\leq  CL\lambda^L + \left(\ep{C_2\frac{g^2}{L-1}}-1\right) + C_3 \abs{g}/L \leq C_1/L.
\end{align*}
for a suitable constant $C_1$.

Now, we consider arbitrary local observables $A,B\in\caA^{[-l^2+1,l^2]}$ and will show 
that
\begin{equation}\label{FCS_U_Limit}
\lim_{L\to\infty}\omega(A\str U_g^+(L)B) = \Tr\left(\rho \bbE^{(2l^2)}_{A\str U_g^+(l)B}(u_g)\right).
\end{equation}
Indeed, for $L>l$,
\begin{align*} 
\omega(A\str U_g^+(L)B) &= \Tr\bigg(\rho \bbE^{(2l^2)}_{A\str U_g^+(l)B}\circ \prod_{x=l^2}^{L^2-1} \bbE_{U_{gf_L(x)}}(\idtyty)\bigg) \\
&= \Tr\bigg(\rho \bbE^{(2l^2)}_{A\str U_g^+(l)B}(u_g)\bigg) + \Tr\bigg(\rho \bbE^{(2l^2)}_{A\str U_g^+(l)B}\circ \bigg[\prod_{x=l^2}^{L^2-1} \bbE_{U_{gf_L(x)}} - Q_g \bigg](\idtyty)\bigg), 
\end{align*}
and the last term is a remainder that is bounded above by $C_1 \Vert A \Vert \Vert B \Vert / (L-l)$. Therefore,
\begin{align*}
\abs{\omega\left(A\str (U_g^+(L)-U_g^+(L'))B\right)} 
&\leq \abs{\omega(A\str U_g^+(L)B) - \Tr\left(\rho \bbE^{(2l^2)}_{A\str U_g^+(l)B}(u_g)\right)} \\
&\quad + \abs{\omega(A\str U_g^+(L')B) - \Tr\left(\rho \bbE^{(2l^2)}_{A\str U_g^+(l)B}(u_g)\right)} \\
&\leq \frac{2C_1 \Vert A \Vert \Vert B \Vert}{\mathrm{min}((L-l),(L'-l))}
\end{align*}
and we obtain the existence of a limiting operator $U_g^+$ by Lemma~\ref{lma:UnitaryCauchy}, with the dense set being
\begin{equation*}
\caH_{loc} := \{\pi_\omega(A)\Omega_\omega:A\in\caA_{loc}\}.
\end{equation*}

For any $g,g'\in (-s,s)$ such that $g+g'\in (-s,s)$, the group property follows from that of the local approximations. Moreover, for any $\psi\in\caH_{loc}$,
\begin{equation*}
\Vert (U_g^+-1)\psi \Vert^2 = \lim_{L\to\infty} \omega\left(A\str (2 - U_g^+(L)\str-U_g^+(L))A\right) =  2 \omega(A\str A) - 2\mathrm{Re\,} \Tr\left(\rho \bbE^{2l^2}_{A\str U_g^+(l) A}(u_g)\right)
\end{equation*}
by~(\ref{FCS_U_Limit}). Letting $g\to 0$, the trace converges to $\Tr\left(\rho \bbE_{A\str A}(\idtyty)\right) = \omega(A\str A)$. Hence, $U_g^+$ is strongly continuous at $g=0$ and by the group property for any $g\in (-s,s)$.

Finally, the last claim of the theorem follows from $\pi_\omega(U_g^+(L))\in\pi_\omega(\caA^{(-\infty,0]})'$, the strong convergence just proved, and the general fact that $\caU' = \caU'''$.
\end{proof}

Since we obtain strong continuity by a perturbative argument and without controlling the convergence of the generators $S^+(L)$, no further 
information on the limiting generator - for example its domain - can be extracted from the proof.

At the heart of the proof is an argument that amounts to a discrete adiabatic theorem for a special class of operators, namely the operators $\bbE_{U_g}$. For unitaries, a discrete
adiabatic theorem has been known for some time \cite{Dranov:1998},  but the $\bbE_{U_g}$ are not unitary and not expected to be normal operators in general. Nevertheless, as
can be seen in the proof, they turn out to posses  all the necessary properties for an adiabatic property to hold.

\subsection{Excess spin in the stochastic-geometric picture}\label{sub:Poisson}

In this section, we use the notations and definitions of~\cite{Nachtergaele:1993wk} and refer to it for a complete exposition, see also~\cite{Ueltschi:2013uy} for recent developments and conjectures. The gauge group is $SU(2)$, and$S^x,S^y,S^z$ are the usual spin matrices in the $d$-dimensional representation of the Lie algebra, where $d = 2S + 1$. Concretely, the Hamiltonians under consideration have the form
\begin{equation}\label{AF_H}
H^\Lambda = -\sum_{\braket{x}{y}\in\Lambda}\sum_{k=0}^{2S} J_k Q_k(S^x\cdot S^y)
\end{equation}
where $\braket{x}{y}\in\Lambda$ denotes the set of pairs of nearest-neighbors in $\Lambda$, $J_k\geq0$ for $0\leq k\leq 2S$, and $Q_k$ is a polynomial of degree $k$ given by
\begin{equation*}
Q_k(z) = 2^k \left[\frac{(2S-k)!}{(2S)!}\right]^2\prod_{l=2S-k+1}^{2S}(z_l-z), \qquad z_l = \frac{l(l+1)}{2}-S(S+1).
\end{equation*}
Note that the a priori intricate form of the polynomials arises from the description of the spin $S$ states as the symmetric subspace of the tensor product of $2S$ spins $1/2$, which interact through a purely antiferromagnetic interaction of the form $(1/4) - (1/2)P^{xy}_{(0)}$, where $P^{xy}_{(0)}$ denotes the projection onto the spin-$0$ subspace of the two spin-$1/2$. In particular, for $S=1$, we have that $Q_1(S^x\cdot S^y) = 1/2(1-S^x\cdot S^y)$. We shall refer to models where $J_{2S}>0$ and all other coupling constants are $0$ as single line models, all others will be called multiline models.

\begin{figure}
\centering
\def\svgwidth{0.5\textwidth}
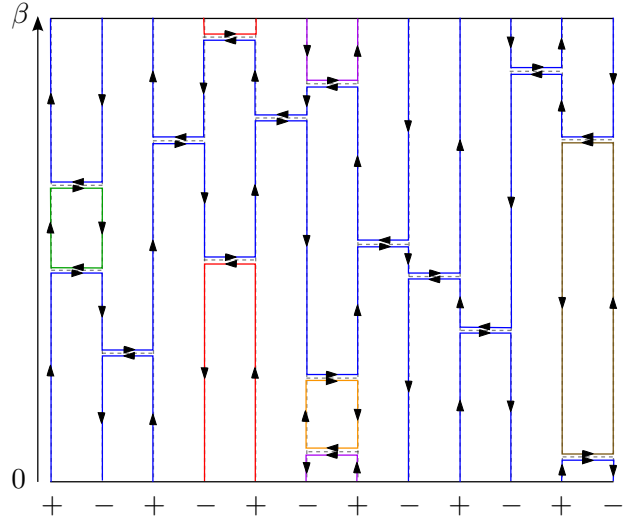
\caption{One particular loop configuration $\nu$ for the single line model on the finite chain $[a,b]$. The distribution of bridges in the space-time $[a,b]\times[0,\beta]$ is given by the Poisson process $\rho_\beta$. At each bridge $(t_i,e_i)$, the projection $P^{e_i}_{(0)}$ imposes the antisymmetry of the spin configuration across $e_i$ both at $t_i^-$ and $t_i^+$. Through $t_i$ and with equal probability, the spin configuration can either remain constant (arrows pointing in the same direction) or the spins at both ends of $e_i$ must flip (arrows pointing in opposite directions).}
\label{fig:loops}
\end{figure}
Crucially, physical quantities such as the spin-spin correlation function have a natural interpretation in the stochastic picture which allows for relatively straightforward computations. For single line models and $J_{2S} = 1$, equilibrium expectation values of observables that are functions of $S^x_3$ can be expressed as
\begin{equation}\label{StochRep}
\omega_\beta\left(f(S^x_3)_{x\in\Gamma}\right) = \int d\mu_\beta(\nu) E_\nu(f),
\end{equation}
where $d\mu_\beta(\nu) = Z_\beta^{-1}d\rho_\beta(\nu) (2S+1)^{l(\nu)}$ and $l(\nu)$ is the number of loops in $\nu$, $d\rho_\beta$ is a Poisson process running at rate $1$ between $t=0$ and $t=\beta$ that determines the position of bridges, see Figure~\ref{fig:loops}, and $E_\nu$ is a linear functional on the algebra of observables which reads
\begin{equation}\label{E(f)}
E_\nu(f) = \frac{1}{(2S+1)^{l(\nu)}} \sum_{\text{Spin configurations }\sigma\text{ consistent with }\nu} f(\sigma(t=0))
\end{equation}
Similar formulas can be derived for general multiline models and for arbitrary observables. Ground state expectation values are obtained in the limit $\beta\to\infty$. In particular,
\begin{equation}\label{EofSpinSpin}
E_\nu(S^{x}_3S^{y}_3) \propto (-1)^{\abs{x-y}}I_\nu(x \sim y)
\end{equation}
where $I_\nu(x \sim y)$ is the indicator function of the event that $(x,0)$ and $(y,0)$ belong to the same loop in the configuration $\nu$, so that
\begin{equation}\label{SpinSpin}
\omega_\beta(S^{x}_3S^{y}_3) \propto (-1)^{\abs{x-y}}\bbP_{\mu_\beta}(x \sim y),
\end{equation}
namely the probability that $(x,0)$ and $(y,0)$ belong to the same loop.

For these models, it is more natural to define the approximants $\hat{S}^+(\epsilon)\in\caA_\Rchain$ and $\hat{S}^-(\epsilon)\in\caA_\Lchain$, for $\epsilon>0$, by
\begin{equation*}
\hat{S}^+(\epsilon) = \sum_{x\geq 1} \ep{-\epsilon x}S^x, \qquad \hat{S}^-(\epsilon) = \sum_{x\leq0} \ep{\epsilon x}S^x
\end{equation*}
and give meaning to the limit as $\epsilon\to 0$. We will treat the case $\hat{S}^+(\epsilon)$ in details, the other one being completely analogous. In the stochastic representation~(\ref{StochRep}, \ref{E(f)}), we observe that the loops that intersect only the positive axis contribute a total spin of zero (because they are closed), so that only the loops encircling $(1/2,0)$ may contribute to the excess spin in the GNS representation of the ground state. This intuition is the object of the following lemma.

\begin{lemma}\label{lma:StochExcessSpin}
For any given configuration $\nu$ of loops $\gamma$, let
\begin{equation}\label{StochExcess}
\hat{S}_\nu := \sum_{(x,k): x\geq1,N_{(1/2,0)}(\gamma(x,k))=1}S^{(x,k)},
\end{equation}
where $S^{(x,k)}$, $k=1\ldots2S$, are the $2S$ spins $1/2$ that make up $S^x$, $\gamma(x,k)$ denotes the loop to which $((x,k),0)$ belongs, and $N_{(1/2,0)}(\gamma)$ is the indicator function of $(1/2,0)\in\mathrm{int}(\gamma)$. If
\begin{equation}\label{thirdmoment}
\sum_{x\in\bbZ}\abs{x^3\,\omega(S^0S^x)}<\infty,
\end{equation}
then
\begin{equation}\label{ConvStochExcess}
\lim_{\epsilon\to0} \int d\mu_\infty(\nu) E_\nu\left[(\hat{S}^+(\epsilon)-\hat{S}^+_\nu)^2\right] = 0.
\end{equation}
\end{lemma}
This convergence is the key to proving that the excess spin operator exists. Note that for ground states $\omega$ with a non-vanishing
spectral gap the assumption (\ref{thirdmoment}) is always satisfied due to the 
Exponential Clustering Theorem~\cite{Nachtergaele:2006aa, Hastings:2006aa}.
\begin{thm}\label{thm:ExcessSpin}
Let $\omega$ be a ground state of a Hamiltonian of the form~(\ref{AF_H}), and let $(\caH_\omega,\pi_\omega,\Omega_\omega)$ be its GNS representation. If
\begin{equation*}
\sum_{x\in\bbZ}\abs{x^3\,\omega(S^0 S^x)}<\infty,
\end{equation*}
then the strong limits
\begin{equation*}
\hat{U}^+_g = \slim_{\epsilon\to 0} \ep{\iu g \cdot \pi_\omega\left(\hat{S}^+(\epsilon)\right)},\qquad 
\hat{U}^-_g = \slim_{\epsilon\to 0} \ep{\iu g \cdot \pi_\omega\left(\hat{S}^-(\epsilon)\right)}
\end{equation*}
exist and define infinite-dimensional, strongly continuous representations of $SU(2)$. Moreover:
\begin{enumerate}
\item If $\caU'$ denotes the commutant of the algebra $\caU$, 
\begin{equation*}
\hat{U}^+_g\in\pi_\omega\left(\caA^\Lchain\right)',\qquad \hat{U}^-_g \in\pi_\omega\left(\caA^\Rchain\right)';
\end{equation*}
\item Let $\hat{S}^+$ and $\hat{S}^-$ be the self-adjoint generators of $\hat{U}^+_g$ and $\hat{U}^-_g$. Then,
\begin{equation*}
\hat{S}^+ = \slim_{\epsilon\to 0} \hat{S}^+(\epsilon),\qquad 
\hat{S}^- = \slim_{\epsilon\to 0} \hat{S}^-(\epsilon);
\end{equation*}
\item The dense space $\caH_{\mathrm{loc}}$ is contained in the domain $\caD\subset\caH_\omega$ common to $\hat{S}^{+}_m$ and $\hat{S}^{-}_m$ for all $0\leq m\leq 2S$.
\end{enumerate}
\end{thm}
\begin{proof}[Proof of Lemma~\ref{lma:StochExcessSpin}]
For simplicity, we denote $\mu_\infty$ by $\mu$ in the rest of the article. First,
\begin{equation*}
\hat{S}^+(\epsilon)-\hat{S}^+_\nu = \sum_{(x,k)}{}^{\!'}(\ep{-\epsilon x}-1)S^{(x,k)}
+ \sum_{(x,k)}{}^{\!''}\ep{-\epsilon x} S^{(x,k)}
\end{equation*}
where $\sum_{(x,k)}{}^{\!'}$ denotes the sum in~(\ref{StochExcess}) and  $\sum_{(x,k)}{}^{\!''} = \sum_{(x,k)} - \sum_{(x,k)}{}^{\!'}$. Since sites in the two sums are by definition not related by any loops,
\begin{equation*}
E_\nu \Big(\sum_{(x,k)}{}^{\!'}(\cdots)\sum_{(x,k)}{}^{\!''}(\cdots)\Big) = E_\nu \Big(\sum_{(x,k)}{}^{\!'}(\cdots)\Big)E_\nu \Big(\sum_{(x,k)}{}^{\!''}(\cdots)\Big),
\end{equation*}
so that
\begin{multline} \label{TwoTerms}
\int d\mu(\nu) E_\nu\left[(\hat{S}^+(\epsilon)-\hat{S}^+_\nu)^2\right] \\
\leq 2\int d\mu(\nu) E_\nu\Big[ \Big(\sum_{(x,k)}{}^{\!'}(\ep{-\epsilon x}-1)S^{(x,k)}\Big)^2\Big] 
+ 2\int d\mu(\nu) E_\nu\Big[ \Big(\sum_{(x,k)}{}^{\!''}\ep{-\epsilon x} S^{(x,k)}\Big)^2\Big]. 
\end{multline}
 We concentrate on the first term. From~(\ref{EofSpinSpin}), we have
\begin{multline*}
\int d\mu(\nu) E_\nu\Big[ \Big(\sum_{(x,k)}{}^{\!'}(\ep{-\epsilon x}-1)S^{(x,k)}\Big)^2\Big] 
\\ \propto \int d\mu(\nu) \sum_{(x,k),(y,l),x,y\geq1} (\ep{-\epsilon x}-1)(\ep{-\epsilon y}-1) I_\nu\big[(x,k)\sim(y,l)\sim(-\infty,0]\big]
\end{multline*}
which converges monotonically to $0$ as $\epsilon\to0$. Therefore, it  suffices to prove that the integrand is uniformly bounded by an integrable function. We apply $\abs{\exp(-\epsilon x)-1}\leq 1$, and then obtain the bound
\begin{multline*}
\int d\mu(\nu) \sum_{(x,k),(y,l),x,y\geq1} I_\nu\big[(x,k)\sim(y,l)\sim(-\infty,0]\big] \\
= 2 (2S)^2 \sum_{1\leq x\leq y} \bbP_\mu\big((x,1)\sim (y,1)\sim(-\infty,0]\big)
\leq 2 (2S)^2 \sum_{1\leq x\leq y} \bbP_\mu\big(x\sim y\sim(-\infty,0]\big),
\end{multline*}
where we used that by the symmetrization of the vertical lines, the sum over the auxiliary indices $k$ and $l$ can be replaced by a factor $(2S)^2$ and fixing $k=l=1$. Since the loops do not intersect,
\begin{equation*} 
\sum_{1\leq x\leq y} \bbP_\mu\big(x\sim y\sim(-\infty,0]\big) \leq \sum_{u<1\leq x\leq y\leq v} \bbP_\mu(u\sim v)
\end{equation*}
and we use translation invariance to obtain
\begin{align*}
\sum_{u<1\leq x\leq y\leq v} \bbP_\mu(u\sim v) = \sum_{d\geq1}\sum_{1\leq x \leq y \leq v\leq d} \bbP_\mu(0\sim d) \,\propto \sum_{d\geq1}\frac{d(d+1)(d+2)}{6}\vert\omega(S^0S^d)\vert<\infty.
\end{align*}
where we recalled~(\ref{SpinSpin}). Hence, the first term of~(\ref{TwoTerms}) converges to zero.

As for the second term, we have
\begin{align*}
\int d\mu(\nu) E_\nu\Big[ \Big(\sum_{(x,k)}{}^{\!''}\ep{-\epsilon x} S^{(x,k)}\Big)^2\Big]
&\propto \int d\mu(\nu)\sum_{(x,k),(y,l)}{}^{\!''}\ep{-\epsilon (x+y)}(-1)^{\abs{x-y}} I_\nu((x,k)\sim(y,l)) \\
& = \int d\mu(\nu)\sum_{\gamma\in\nu_+}\bigg(\sum_{(x,k)\in\gamma}\ep{-\epsilon x}(-1)^{x}\bigg)^2
\end{align*}
where $\nu_+\subset\nu$ is the subset of loops that intersect $t=0$ only on $[1,\infty)$, we used~(\ref{EofSpinSpin}) in the first equality, and $(-1)^{\abs{x-y}} = (-1)^{(x+y)}$. Now, this formally vanishes at $\epsilon = 0$ by the fact that all loops are closed, but it remains to be seen that the decay of the spin-spin correlation function implies a sufficient decay in the probability that the loops are large. For any $\gamma\in\nu_+$, let $d_\gamma$ be the largest distance between intersections of $\gamma$ and the $t=0$ axis, and $n_\gamma$ the number of such intersections. If $((y,l),0)\in\gamma$ is the leftmost point of $\gamma$, we have the bound
\begin{equation*}
\bigg\vert \sum_{(x,k)\in\gamma}\ep{-\epsilon x}(-1)^{x} \bigg\vert
\leq \ep{-\epsilon y}\frac{n_\gamma}{2}\left(1-\ep{-\epsilon d_\gamma}\right)
\end{equation*}
as the even number of intersections can be paired so that the two sites of each pair belong to different sublattices. Hence,
\begin{equation*}
\int d\mu(\nu) E_\nu\Big[ \Big(\sum_{(x,k)}{}^{\!''}\ep{-\epsilon x} S^{(x,k)}\Big)^2\Big]
\leq \sum_{y\geq1} \frac{\ep{-2\epsilon y}}{4} \int d\mu(\nu) \sum_{k=1}^{2S}\frac{1}{n_{\gamma(y,k)}}n_{\gamma(y,k)}^2\left(1-\ep{-\epsilon d_{\gamma(y,k)}}\right)^2.
\end{equation*}
Now, we use the simple bound $n_\gamma\leq 2S d_\gamma$, the equidistribution over $k$ and translation invariance to conclude that
\begin{align*}
\int d\mu(\nu) E_\nu\Big[ \Big(\sum_{(x,k)}{}^{\!''}\ep{-\epsilon x} S^{(x,k)}\Big)^2\Big]
&\leq \frac{(2S)^2}{4 (\ep{2\epsilon}-1)}  \int d\mu(\nu) d_{\gamma(0,1)}\left(1-\ep{-\epsilon d_{\gamma(0,1)}}\right)^2 \\
&\leq \frac{(2S)^2}{4 (\ep{2\epsilon}-1)}  \int d\mu(\nu) \sum_{x\geq 0}x \left(1-\ep{-\epsilon x}\right)^2 I_\nu(0\sim x) \\
&\leq \frac{(2S)^2}{4 (\ep{2\epsilon}-1)} \sum_{x\geq 0}x \left(1-\ep{-\epsilon x}\right)^2 \bbP_\mu(0\sim x).
\end{align*}
Finally, $1-\ep{-\epsilon x}\leq \epsilon x$ and $\epsilon^2 / (\ep{2\epsilon}-1)\leq \epsilon$ so that
\begin{equation*}
\int d\mu(\nu) E_\nu\Big[ \Big(\sum_{(x,k)}{}^{\!''}\ep{-\epsilon x} S^{(x,k)}\Big)^2\Big]
\leq \epsilon S^2 \sum_{x\geq0} x^3 \, \vert\omega(S^0S^x)\vert.
\end{equation*}
Hence, the second term also vanishes as $\epsilon\to0$.
\end{proof}

Before we go into the proof proper of the theorem, we need the following technical lemma.

\begin{lemma}\label{lma:Cauchy_Gen}
Let $(U_n(\cdot))_{n\in\bbN}$ be a sequence of unitary groups on $\caH$ generated by bounded self-adjoint operators $(A_n)_{n\in\bbN}$ such that
\begin{equation}\label{Cauchy_Gen}
\lim_{\min(n,m)\to\infty} \braket{\phi}{(A_n-A_m)^2\psi} = 0
\end{equation}
for all $\phi,\psi$ in a dense subset $\caD\subset\caH$. Then there exists a strongly continuous unitary group $U(t) = \exp(\iu A t)$, with $\caD \subset \mathrm{Dom}(A)$, such that
\begin{equation*}
\slim_{n\to\infty} U_n(t) = U(t), \qquad \text{and}\qquad \lim_{n\to\infty} A_n\phi = A\phi,\quad \phi\in\caD;
\end{equation*}
\end{lemma}
\begin{proof}
First, for any fixed $t$ and $\phi,\psi\in\caD$, we have that
\begin{align} \label{FirstConv}
\abs{\braket{\phi}{(1-U_n(t)\str U_m(t))\psi}}^2
&\leq 2\Vert\phi\Vert^2\braket{\psi}{[1-\cos((A_n-A_m)t)]\psi} \\
&\leq 2t^2\Vert\phi\Vert^2 \braket{\psi}{(A_n-A_m)^2\psi} \nonumber
\end{align}
where we used Schwarz' inequality and then the functional calculus. Hence, (\ref{Cauchy_Gen}) implies that $\min_{n,m}\braket{\phi}{U_n(t)\str U_m(t)\psi}$ converges to $\braket{\phi}{\psi}$, and similarly for $\min_{n,m}\braket{\phi}{U_m(t)\str U_n(t)\psi}$, which yield the existence of a limiting unitary $U(t)$ by Lemma~\ref{lma:UnitaryCauchy}.

The convergence of $A_n\psi$ follows from $\Vert(A_n-A_m)\psi\Vert^2 = \braket{\psi}{(A_n-A_m)^2\psi}$ and~(\ref{Cauchy_Gen}) by completeness. Finally, the group property of $U(t)$ is inherited from $U_n(t)$, and
\begin{equation*}
\frac{1}{t}(1-U(t))\psi - A\psi = \left[\frac{1}{t}(1-U_n(t)) - A_n\right]\psi + (A_n-A)\psi - \frac{1}{t}(U-U_n(t))\psi 
\end{equation*}
where the last term converges to zero as $(n\to\infty)$, uniformly in $t$ by~(\ref{FirstConv}). This yields the strong differentiability on $\caD$, with $A$ being the generator of $U(t)$. 
\end{proof}
\begin{proof}[Proof of Theorem~\ref{thm:ExcessSpin}]
By Lemma~\ref{lma:Cauchy_Gen}, it suffices to prove that
\begin{equation*}
\lim_{\max(\epsilon,\epsilon')\to0}\omega\left(A\str(\hat{S}^+(\epsilon)-\hat{S}^+(\epsilon'))^2B\right) = 0,
\end{equation*}
which can be reduced further to the case $A=B$ by Cauchy-Schwarz' inequality. Moreover,
\begin{align*}
\omega\left(B\str(\hat{S}^+(\epsilon)-\hat{S}^+(\epsilon'))^2B\right) \leq 2\omega\left(B\str(\hat{S}^{[1,L)}(\epsilon)-\hat{S}^{[1,L)}(\epsilon'))^2B\right) + 2\omega\left(B\str(\hat{S}^L(\epsilon)-\hat{S}^L(\epsilon'))^2B\right),
\end{align*}
where we used that $\omega(B\str(X-Y)^2B)\geq 0$ implies that $\omega(B\str(XY+YX)B)\leq\omega(B\str(X^2+Y^2)B)$ for any observable $X,Y$. The first term converges to zero by the continuity of $\omega$ and the norm convergence of $\hat{S}^{[1,L)}$ for $L<\infty$. In the stochastic representation, the second term reads
\begin{equation*}
\omega\left(B\str(\hat{S}^L(\epsilon)-\hat{S}^L(\epsilon'))^2B\right) = \int d\mu(\nu)E_\nu\left(B\str(\hat{S}^L(\epsilon)-\hat{S}^L(\epsilon'))^2B\right)
\end{equation*}
which approximately factorizes as $\int d\mu(\nu)E_\nu\left(B\str B\right)E_\nu\left((\hat{S}^L(\epsilon)-\hat{S}^L(\epsilon'))^2\right)$. Indeed, the difference is given by the loops that connects the finite support of $B\str B$ and $[L,\infty)$. Each such loop $\gamma$ contributes at most $n_\gamma$ terms from $(\hat{S}^L(\epsilon)-\hat{S}^L(\epsilon'))^2$, so that
\begin{multline*}
\Bigg\vert\int d\mu(\nu) E_\nu\left(B\str(\hat{S}^L(\epsilon)-\hat{S}^L(\epsilon'))^2B\right)-\int d\mu(\nu)E_\nu\left(B\str B\right)E_\nu\left((\hat{S}^L(\epsilon)-\hat{S}^L(\epsilon'))^2\right)\Bigg\vert \\
\leq 2S K_1(B) \int d\mu(\nu) \sum_{\gamma\in\nu} d_\gamma \, I\left(\mathrm{supp}(B)\stackrel{\gamma}{\sim}[L,\infty)\right) \leq K_2(B) \sum_{x\geq L-l} x \vert\omega(S^0S^x)\vert,
\end{multline*}
where $\mathrm{supp}(B) \in \caA^{(-\infty,l]}$ and we assumed $l<L$. This is arbitrarily small for $L$ large enough. Therefore, it suffices to consider the case $B=\idtyty$. But this is an immediate consequence of the inequality
\begin{equation*}
\int d\mu(\nu) E_\nu\left[(\hat{S}^+(\epsilon)-\hat{S}^+(\epsilon'))^2\right] \leq 2\int d\mu(\nu) E_\nu\left[(\hat{S}^+(\epsilon)-\hat{S}^+_\nu)^2\right] + 2 \int d\mu(\nu) E_\nu\left[(\hat{S}^+(\epsilon')-\hat{S}^+_\nu)^2\right]
\end{equation*}
that follows again from the positivity of $\omega((\hat{S}^+(\epsilon)-\hat{S}^+(\epsilon'))^2)$, and Lemma~\ref{lma:StochExcessSpin}.
\end{proof}
%


\section{Boundary operators, edge representations and correlation structure}\label{sec:boundary-bulk}

Now, we put together the various elements introduced in the previous sections and clarify the relationship between the boundary operators of Section~\ref{sec:Excess} and the abstract representations of $G$ at the edges of a symmetric gapped ground state phase, see Section~\ref{sec:GTransport}. We find an exact relation for frustration free systems. This allows us to show that the representation content of the boundary operator restricted to the kernel of the half-chain is an invariant of a $G$-symmetric phase. For simplicity, we continue to assume the uniqueness and therefore translation invariance of the ground state in the thermodynamic limit. 

We start by interpreting the boundary operator arising for finitely correlated systems in the light of a $G$ action. 
\begin{lemma}\label{lma:Implementable}
Assume that the conditions of Theorem~\ref{thm:ExcessSpin} are satisfied. Then $\theta_g^\Rchain$ and $\theta_g^\Lchain$ are unitarily implementable on the GNS Hilbert space $\caH_\omega$.
\end{lemma}
\begin{proof}
We consider the right infinite chain. Theorem~\ref{thm:ExcessSpin} implies the existence of the excess spin operator. First, we claim that the representation $U^\Rchain_g$ implements $\theta^\Rchain_g$ on the local algebra $\cup_{\Lambda\subset\Rchain}\caA^\Lambda$. Indeed, for any local $A$, and vectors $\psi,\phi\in\caH_\omega$, we have
\begin{align*}
\braket{\psi}{U^{\Rchain\,*}_{g}\pi_\omega(A)U^\Rchain_{g}\phi}
&=\lim_{\epsilon,\epsilon'\to0}\braket{\ep{\iu g\cdot\pi_\omega(S^+(\epsilon))}\psi}{\pi_\omega(A)\ep{\iu g\cdot\pi_\omega(S^+(\epsilon'))}\phi} \\
&= \lim_{\epsilon,\epsilon'\to0}\braket{\psi}{\pi_\omega\left(\ep{-\iu g\cdot S^+(\epsilon)}A\ep{\iu g\cdot S^+(\epsilon')}\right)\phi} \\
&= \braket{\psi}{\pi_\omega\left(\theta^\Rchain_g(A)\right)\phi}
\end{align*}
where we used in the last equality that the GNS representation is continuous and that $A$ has finite support. We conclude by density and again the continuity of $\pi_\omega$ and of $\theta_g$ that this extends to $\caA^\Rchain$.
\end{proof}

In the remainder of this section we restrict our attention to quantum spin chains with a frustration-free, translation invariant
finite-range interaction. By this, we mean a  system with a translation invariant interaction 
with ground state space $\caG^\Lambda$ which satisfies the following intersection property: There exists $l\geq1$ such that 
\begin{equation*}
\caG^{[1,L]} = \bigcap_{k=1}^{L-l}\otimes_{x=1}^{k-1}\caH^x\otimes\caG^{[k,k+l]}\otimes_{y=k+l+1}^L\caH^y
\end{equation*}
for all $L\geq l$. In other words, the global energy minimizer also minimizes the energy locally.
In this case, all ground states have exactly the same energy, and furthermore,
\begin{equation*}
\omega\in\caS^\bbZ\:\Longrightarrow\:
\omega\upharpoonright_{\caA^\Rchain}\in\caS^\Rchain\text{ and }
\omega\upharpoonright_{\caA^\Lchain}\in\caS^\Lchain.
\end{equation*}
This yields the following immediate corollary that relates the concrete boundary operator to the abstract representation of $G$ on the set of edge states.
\begin{prop}\label{prop:FFGS}
The two orbits
\begin{equation*}
\caK:= \{U_g^+\Omega_\omega:g\in G\}\subset\caH_\omega,\qquad\caV:=\{\Theta^{\Rchain}_g\circ\omega\upharpoonright_{\caA^\Rchain}:g\in G\}\subset \caS^\Rchain
\end{equation*}
carry equivalent representations of $G$ and are subrepresentations of $\caS^\Rchain$.
\end{prop}
\begin{proof}
The isomorphism $\phi:\caV\rightarrow\caK$ is given by $\phi(\Theta^{\Rchain}_g\circ\omega) = U_g^+\Omega_\omega$. Indeed, any vector $\nu\in\caV$ is given by $\nu = \Theta^{\Rchain}_g\circ\omega$ for a $g\in G$ and
\begin{equation*}
\phi(\Theta^{\Rchain}_{g'}\circ\nu) = \phi(\Theta^{\Rchain}_{g'}\circ\Theta^{\Rchain}_g\circ\omega) = U_{g'g}^+\Omega_\omega = U_{g'}\phi(\nu).
\end{equation*}
By construction, $\caV$ is a subrepresentation of $\caS^\Rchain$.
\end{proof}

In the following, we furthermore assume that ground states are finitely correlated, i.e., are of matrix product form. 
We recall that, in this situation, the $G$-invariance is expressed by two representations $U_g$ and $u_g$ in
the special unitary groups of $\caA^{0}$, and $\caB$, respectively, with an intertwiner $V$ with the properties discussed 
before (see (\ref{intertwine}) and (\ref{FCS_Spectrum})). We also recall that for any  $\omega\in\caS^{[1,\infty)}$ there is density 
matrix  $\sigma_\omega \in \caB$ such that for all $A\in \caA_{loc}$ we have
\begin{equation}\label{density_matrix}
\omega(A) =\Tr \left(\sigma_\omega \bbE_A (\idtyty)\right).
\end{equation}
Not surprisingly, there is an intimate relation between the boundary representation $U_g^+$, the half-chain edge representation $\Theta_g^{[1,\infty)}$, and the auxiliary representation $u_g$ used in the construction of the matrix product state state.

In the GNS representation of the ground state on the infinite chain ($\Gamma = \bbZ$), the set of vectors
\begin{equation*}
\caG^+:=\overline{\{\pi_\omega(A)\Omega_\omega : A\in\caA_{loc}, \mathrm{supp}(A)\subset (-\infty,0]\}}
\end{equation*}
represents the ground state space of the Hamiltonian on the right half-infinite chain. Since $U_g^+\in\pi_\omega(\caA^{(-\infty,0]})'$ and $\Omega_\omega$ is $G$-invariant, we have that
\begin{equation*}
U_g^+\pi_\omega(A)\Omega_\omega = \pi_\omega(A)(U_g^-)\str\Omega_\omega = (U_g^-)\str\pi_\omega(\theta_g^{-1}(A))\Omega_\omega \in \caG^+,
\end{equation*}
so that $\caG^+$ is invariant under the action of $U_g^+$, and the infinite-dimensional representation $U_g^+\upharpoonright_{\caG^+}$ is again that of the action of $G$ on the right half-infinite chain. 
\begin{thm} \label{thm:FF}
1. $U_g^+\upharpoonright_{\caG^+}$ contains only the finite-dimensional representation $u_g$, with infinite multiplicity. \\
2. In the matrix product representation~(\ref{density_matrix}), the action of $\Theta_g^{[1,\infty)}$ on $\omega\upharpoonright_{\caA^\Rchain}$ is implemented by ${\rm Ad}_{u_g\str}$ on~$\sigma_\omega$.
\end{thm}
\begin{proof}
1. As a corollary of the proof of Theorem~\ref{thm:ExcessSpin_FCS}, we have that for any $A_1,A_2\in\caA^\Lchain$
\begin{equation}\label{ExcessSpinEquivalence}
\braket{\pi_\omega(A_1)\Omega_\omega}{U_g^+\pi_\omega(A_2)\Omega_\omega} = \braket{\Omega_\omega}{\pi_\omega(A_1\str A_2)U_g^+\Omega_\omega} = \Tr\left(\rho\bbE_{A_1\str A_2}(u_g)\right),
\end{equation}
where the second equality follows from~(\ref{FCS_U_Limit}). Since in the minimal representation of the finitely correlated state,
\begin{equation*}
\mathrm{span}\left\{\rho \circ \bbE_{A_1}\circ\cdots \circ\bbE_{A_n}: n\in\bbN, A_i\in\caA^0 \right\} = \caB\str,
\end{equation*}
equation~(\ref{ExcessSpinEquivalence}) gives a one-to-one correspondence between matrix entries of $U_g^+\upharpoonright_{\caG^+}$ and those of $u_g$. Hence, $u_g$ determines the full set
of representation content of $U_g^+\upharpoonright_{\caG^+}$. As $U_g^+\upharpoonright_{\caG^+}$ is infinite-dimensional, the finite-dimensional representation $u_g$ must appear in it with infinite degeneracy.

\noindent 2. We have
\begin{align*}
\Theta_g^{\Rchain}(\omega)(A) &= \omega(\theta_g^{\Rchain}(A)) = \Tr \left(\sigma_\omega \bbE_{\theta_g^{\Rchain}(A)}(\idtyty)\right)
= \Tr \left(\sigma_\omega \mathrm{Ad}_{u_g} (\bbE_A(\idtyty))\right) \\
&= \Tr \left(\mathrm{Ad}_{u_g^*} (\sigma_\omega) \bbE_A(\idtyty)\right)
\end{align*}
for any $g\in G$.
\end{proof}

We conclude this section with some remarks. First, since $U_g^+$ is a quasi-local observable, $U_g^+\upharpoonright_{\caG^+}$ can be determined from the bulk ground state of the model, computationally through $u_g$ and, in principle, also experimentally. Secondly, we know from the general consideration in Section~\ref{sec:GTransport} that  the maps $\Theta_g^{\Rchain}$ are all equivalent representations within a class of automorphically equivalent $G$-invariant models, so that the same must hold for $u_g$ and the representation content of $U_g^+\upharpoonright_{\caG^+}$. Therefore, within the context of frustration free models with matrix product ground states with a symmetry 
represented by a Lie subgroup of the special unitary groups, Theorem~\ref{thm:FF} shows that not only is the edge representation of the 
symmetry group an invariant for symmetric gapped ground state phases in one dimension, but also that it can be determined, and 
in principle also experimentally observed, from the correlation structure in the bulk. This is what is sometimes referred to as 
symmetry-protected topological order, see \eg~\cite{Chen:2013}.

As an example, let us consider the spin-$1$, $SU(2)$-invariant chains in the Haldane phase, of which the AKLT model~\cite{Affleck:1988vr} is a representative, and whose ground state is finitely correlated~\cite{Fannes:1992vq}. For the AKLT model, the auxiliary representation $u_g$ is the spin-$1/2$ representation, $\caD_{1/2}$, of $SU(2)$. This is also obvious in the valence bond picture of ~\cite{Affleck:1988vr}. Hence, the excess spin observed by rotating the right half-infinite chain, conditioned arbitrarily on the left chain as in~(\ref{ExcessSpinEquivalence}) is also a spin-$1/2$. This is precisely the spin-$1/2$ representation that has been experimentally 
observed in~\cite{Hagiwara:1990wk}. 

It is known that the ground state of the AKLT model can be related to a model with trivial product state 
by a quasi-local automorphism \cite{Schuch:2011ve, Bachmann:2012uu}, but this can obviously not be done 
without breaking the $SU(2)$ symmetry. 
Theorem~\ref{thm:FF} implies the stronger statement that two $SU(2)$-invariant models can be thus connected to each other, without
breaking the $SU(2)$ symmetry or closing the gap, only if their boundary $SU(2)$ representations agree.


\section*{acknowledgments}
B.N. acknowledges stimulating discussions with M. Aizenman about excess charges, spin, and all that,
dating back to the collaboration \cite{Aizenman:1994th}, which turn out to be of great relevance in the context of this work. 
This research was supported in part by the National Science Foundation: S.B. under Grant \#DMS-0757581 and B.N. 
under grant \#DMS-1009502


\end{document}